\documentclass{llncs}
\usepackage{amsmath,amssymb,amsfonts,mathrsfs,thm-restate,thmtools}

\usepackage{graphicx}
\usepackage{appendix}
\usepackage{epsf}
\usepackage{epsfig}
\usepackage{epstopdf}
\usepackage{xspace,xcolor}
\usepackage{soul}
\usepackage{latexsym}
\usepackage{tikz}
\usepackage{framed}
\usepackage{todonotes}
\usepackage{lipsum}
\usepackage{setspace}
\usepackage{xcolor}
\usepackage{tabulary}
\usepackage{tabularx}
\usepackage{multirow}
\usepackage{booktabs}

\usepackage{caption}
\usepackage{subcaption}
\usepackage{tkz-tab}
\usetikzlibrary{arrows,positioning}
\usetikzlibrary{shapes,snakes}

\usepackage[bookmarks=false, colorlinks,
linkcolor=red!60!black, citecolor=green!60!black]{hyperref}
\usepackage{nameref}

\usepackage{enumerate}
\usepackage{caption}
\usepackage{subcaption}
\usepackage[compatibility=false]{caption}
\tikzstyle{vertex}=[circle, draw, inner sep=0pt, minimum size=4.5pt]

\newcommand{\NP}{\ensuremath{\sf{NP}}\xspace}

\newcommand{\fpt}{\ensuremath{\sf{FPT}}\xspace}

\newcommand{\lcp}{\textsc{$2$-Load Coloring}}

\title{On Structural Parameterizations of  Load Coloring}

\author{
	\inst{}  
	I. Vinod Reddy \inst{}
}
\institute{Department of Electrical Engineerng and Computer Science\\
	Indian Institute of Technology Bhilai, Raipur, India \\
	\email{vinod@iitbhilai.ac.in},
} 

\begin{document}
	\pagestyle{plain}
	
	\maketitle
 \begin{abstract}
  Given a graph $G$ and a positive integer $k$, the \lcp{} problem is to check whether there is a $2$-coloring $f:V(G) \rightarrow \{r,b\}$ of $G$ such that for every $i \in \{r,b\}$, there are at least $k$ edges with both end vertices colored $i$. It is known that the problem is $\NP$-complete even on special classes of graphs like regular graphs. Gutin and Jones (Inf Process Lett 114:446-449, 2014) showed that the problem is fixed-parameter tractable by giving a kernel with at most $7k$ vertices. Barbero et al. (Algorithmica 79:211-229, 2017) obtained a kernel with less than $4k$ vertices and $O(k)$ edges, improving the earlier result.
  
  In this paper, we study the parameterized complexity of the problem with respect to structural graph parameters. We show that \lcp{} cannot be solved in time $f(w)n^{o(w)}$, unless ETH fails and it can be solved in time $n^{O(w)}$, where $n$ is the size of the input graph, $w$ is the clique-width of the graph and $f$ is an arbitrary function of $w$.  Next, we consider the parameters distance to cluster graphs, distance to co-cluster graphs and distance to threshold graphs, which are weaker than the parameter clique-width and show that the problem is fixed-parameter tractable (FPT) with respect to these parameters. Finally, we show that \lcp{} is $\NP$-complete even on bipartite graphs and split graphs.
\end{abstract}


\section{Introduction}\label{S:intro}
Given a graph $G$ and a positive integer $c$, the load distribution of a $c$-coloring $f:V(G) \rightarrow [c]$ is a tuple $(f_1,\ldots,f_c)$, where $f_i$ is the number of edges with at least one end point colored with $i$. 
The $c$-\textsc {Load Coloring} problem is to find a coloring $f$ such that the function $\ell_f(G)= \max\{f_i : i \in [c]\}$ is minimum. We denote this minimum by $\ell(G)$. Ahuja et al~\cite{ahuja2007minimum} showed that the problem is $\NP$-hard on general graphs when $c=2$. They also gave a polynomial time algorithm for \lcp{} on trees.  

In a $2$-coloring $f: V(G) \rightarrow \{r,b\}$, an edge is called red (resp. blue) if both end vertices are colored with $r$ (resp. $b$). We use $r_f$ and $b_f$ to denote the number of red and blue edges in a $2$-coloring $f$ of $G$. 
Let $\mu_f(G)=\min\{r_f, b_f\}$ and $\mu(G)$ is the maximum of $\mu_f(G)$ over all possible $2$-colorings of $G$. Ahuja et al.~\cite{ahuja2007minimum} showed that the \lcp{} problem is equivalent to maximizing $\mu(G)$ over all possible $2$-colorings of $G$, in particular they showed that $\ell(G)=|E(G)|-\mu(G)$.

\begin{framed}

	\noindent \textsc{} \\
	\textbf{Input:} A graph $G=(V,E)$ and an integer $k$ \\
	\hspace{-0.1cm} \textbf{Question:}  Does there exists a coloring $f: V(G) \rightarrow \{r, b\}$ such that $\mu(G) \geq k$? (i.e, $r_f \geq k$ and $b_f \geq k$)

\end{framed}

The above version of the load coloring problem has been studied from the parameterized complexity perspective. Gutin et al.~\cite{gutin2014parameterized} proved that the problem admits a polynomial kernel (with at most $7k$ vertices) parameterized by $k$. They also showed that the problem is fixed-parameter tractable when parameterized by the tree-width of the input graph. More recently, Barbero et al.~\cite{barbero2017parameterized}  obtained a kernel for the problem with at most $4k$ vertices improving the result of~\cite{gutin2014parameterized}. 

In this paper, we study the following variant of the load coloring problem.

\begin{framed}

	\noindent \textsc{\lcp{}} \\
	\textbf{Input:} A graph $G=(V,E)$ and integers $k_1$ and $k_2$ \\
	\textbf{Question:}  Does there exists a coloring $f: V(G) \rightarrow \{r, b\}$ such that $r_f \geq k_1$ and $b_f\geq k_2$?

\end{framed}

\paragraph{Our contributions.}
In this paper, we study the \lcp{} problem from the viewpoint of parameterized complexity. A parameterized problem with input size $n$ and parameter $k$ is called fixed-parameter tractable ($\fpt$) if it can be solved in time $f(k)n^{O(1)}$, where $f$ is a function only depending on the parameter $k$ (for more details on parameterized complexity refer to the text~\cite{CyganFKLMPPS15}).
There are many possible parameterizations for \lcp{}. One such parameter is the size of the solution. The problem admits a linear kernel~\cite{gutin2014parameterized} with respect to the size of the solution. 
In this paper, we study the \lcp{} problem with respect to various structural graph parameters. These parameters measure the complexity of the input rather than the problem itself. Tree-width is one of the well-known structural graph parameters. The \lcp{} problem is $\fpt$ when parameterized by tree-width~\cite{gutin2014parameterized} of the input graph. 

Even though tree-width is a widely used graph parameter for sparse graphs, it is not suitable for dense graphs, even if they have a simple structure. In Section~\ref{sec-cw}, we consider the graph parameter clique-width introduced by Courcelle and Olariu~\cite{courcelle2000upper}, which is a generalization of the parameter tree-width. We show that \lcp{} can be solved in time $n^{O(w)}$, and cannot be solved in $f(w) n^{o(w)}$ unless ETH fails, where $w$ is the clique-width of the $n$ vertex input graph and $f$ is an arbitrary function of $w$.  Next, we consider the parameters  distance to cluster graphs, distance to co-cluster graphs, and distance to threshold graphs. These parameters are weaker than the parameter clique-width in the sense that they are a subclass of bounded clique-width graphs. Thus studying the parameterized complexity of \lcp{} with respect to these parameters reduces the gap between tractable and intractable parameterizations. In Section~\ref{sec-palgo}, we show that \lcp{} is fixed-parameter tractable with respect to the parameters distance to cluster, distance to co-cluster  and distance to threshold graphs. Finally in Section~\ref{sec-bipartite}, we show that 
\lcp{} is $\NP$-complete on bipartite graphs and split graphs. Table~\ref{table:results} gives an overview of our results.



%

\begin{table}[t]
	\centering\footnotesize
	\caption{: Known and new parameterized results for \lcp{}}
	\begin{tabular}{ p{.03\textwidth} p{.40\textwidth} p{.50\textwidth}} 
		\toprule
		
	& \textbf{Parameter} & \textbf{Results} \\%

		\midrule
	
	&	size of the solution  &
		Linear kernel~\cite{barbero2017parameterized,gutin2014parameterized} \\ \hline
	&	tree-width  & FPT~\cite{gutin2014parameterized} \\ \hline
		
	&	clique-width($w$)  & $n^{O(w)}$ algorithm  (\autoref{thm-cw-xp}) \\
	& & no $f(w)n^{o(w)}$ algorithm (\autoref{thm-cw-lower})\\ \hline
	
	&	distance to cluster graphs  & FPT (\autoref{thm-dist-cluster-1}) \\ \hline
	&	distance to co-cluster graphs & FPT (\autoref{thm-dist-cocluster-1}) \\ \hline
	&	distance to threshold graphs  & FPT (\autoref{thm-dist-threshold}) \\ \hline
	&	bipartite graphs  & para-NP-hard (\autoref{thm-bipartite}) \\ \hline
	&	split graphs  & para-NP-hard (\autoref{thm-split}) \\ \hline

	\end{tabular}
	\label{table:results}
\end{table}

\section{Preliminaries}
In this section, we introduce some basic notation and terminology related to graph theory and parameterized complexity. 
For $k \in \mathbb{N}$, we use $[k]$ to denote the set $\{1,2,\ldots,k\}$. 
If $f: A \rightarrow B$ is a function and $C \subseteq A$, $f|_C$ denotes the restriction of $f$ to $C$, that is $f|_C: C \rightarrow B$ such that for all $x \in C$, $f|_C(x)=f(x)$
All graphs we consider in this paper are undirected, connected, finite and simple. For a  graph $G=(V,E)$, by $V(G)$ and $E(G)$ we denote the vertex set and edge set of $G$ respectively. We use $n$ to denote the number of vertices and $m$ to denote the number of edges of a graph.
An edge between vertices $x$ and $y$ is denoted as $xy$ for simplicity. 
For a  subset $X \subseteq V(G)$, the graph $G[X]$ denotes the subgraph of $G$ induced by vertices of $X$ and $E_G[X]$ denote the set of edges having both end vertices in the set $X$. For subsets $X, Y \subseteq V(G)$, $E_G[X,Y]$ denote the set of edges connecting $X$ and $Y$.

For a vertex set $X \subseteq V(G)$, we denote $G - X$, the graph obtained from $G$ by deleting all vertices of $X$ and their incident edges. 
For a vertex $v\in V(G)$,
by $N(v)$, we denote the set $\{u \in V(G)~|~ vu \in E(G)\}$ and  we use $N[v]$ to denote the set $N(v) \cup \{v\}$. The neighbourhood of a vertex set $S \subseteq V(G)$ is $N(S)=(\cup_{v \in V(G)} N(v)) \setminus S$.
A vertex is called \emph{universal vertex} if it is adjacent to every other vertex of the graph.  For more details on standard graph-theoretic notation and terminology, we refer the reader  to the text \cite{diestel2005graph}.

%

\subsection{Graph classes} We now define the graph classes which are considered in this paper. 
A graph is \emph{bipartite} if its vertex set can be partitioned into two disjoint sets such that no two vertices in the same set are adjacent.
A \emph{cluster} graph is a disjoint union of complete graphs.  A \emph{co-cluster} graph is the complement graph of a cluster graph.
A graph is a \emph {split graph} if its vertices can be partitioned into a clique and an independent set.  Split graphs are $(C_4,C_5,2K_2)$-free.
A graph is a \emph{threshold graph} if it can be constructed from the one-vertex graph by repeatedly adding either an isolated vertex or a universal vertex. 
The class of threshold graphs is the intersection of split graphs and cographs~\cite{mahadev1995threshold}. Threshold graphs are $(P_4, C_4, 2K_2)$-free.
We denote a split graph (resp. threshold graph) with $G = (C,I)$ where $C$ and $I$ denotes the partition of $G$ into a clique and an independent set.

 For a graph class $\mathcal{F}$  the distance to $\mathcal{F}$ of a graph  $G$ is the minimum number of vertices that have to be deleted from $G$ in order to obtain a graph in $\mathcal{F}$. The parameters distance to cluster graphs~\cite{guo2009more}, distance to co-cluster graphs, distance to threshold graphs~\cite{cai1996fixed} can be computed in \fpt time.

\subsection{Clique-width}
The \emph{clique-width} of a graph $G$  denoted by $cw(G)$, is defined as the minimum number of labels needed to construct $G$ 
using the following four operations:
\begin{enumerate}
	\setlength{\itemsep}{1pt}
	\setlength{\parskip}{0pt}
	\item [i.] \emph{Introducing a vertex.} $\Phi=v(i)$, creates a new vertex $v$ with label $i$. $G_{\Phi}$ is a graph consisting a single vertex $v$ with label $i$. 
	\item [ii.] \emph{Disjoint union.} $\Phi=\Phi' \oplus \Phi''$,  $G_{\Phi}$ is a disjoint union of labeled graphs $G_{\Phi'}$ and $G_{\Phi''}$
	\item [iii.]\emph{Introducing edges.} $\Phi=\eta_{i,j}(\Phi')$, connects each vertex with label $i$ to each vertex with label $j$ ($i \neq j$) in $G_{\Phi'}$.
	\item [iv.]  \emph{Renaming labels.} {\bf $\Phi= \rho_{i\rightarrow j}(\Phi')$}: each vertex of label $i$ is changed to label $j$ in $G_{\Phi'}$.
\end{enumerate}
An  expression build from the above four operations using $w$ labels is called as $w$-$expression$. 
In otherwords, the {\it clique-width} of a graph $G$, is the minimum $w$ for which there exists a $w$-expression that defines the graph $G$. 
A $w$-expression $\Psi$ is a \emph{nice} $w$-expression of $G$, if no edge is introduced twice in $\Psi$.

\section{Graphs of Bounded Clique-width}\label{sec-cw}

\subsection{Upper bound}
In this section, we present an algorithm for solving \lcp{} which runs in time $n^{O(w)}$ on graphs of clique-width at most $w$.
\begin{theorem}\label{thm-cw-xp}
	\lcp{} can be solved in time $n^{O(w)}$, where $w$ is the clique-width of the input graph.
\end{theorem}
\begin{proof}
	The algorithm is based on  a dynamic programming over the $w$-expression of the input graph $G$. We assume that the $w$-expression $\Psi$ defining $G$ is nice, that is every edge is introduced exactly once in $\Psi$. 
	
	For each subexpression $\Phi$ of $\Psi$, 
	
	$$OPT(\Phi, n_{1,r}, n_{1,b},n_{2,r}, n_{2,b},\ldots, n_{w,r}, n_{w,b},k_1)$$
	
	denotes maximum number of blue edges that can be obtained in a $2$-coloring $f:V(G_{\Phi}) \rightarrow \{r,b\}$ of $G_{\Phi}$ with the constraint that number of red edges is at least $k_1$ in $G_{\Phi}$ and the number of vertices of label $i$ in $G_{\Phi}$ that are colored with a color $\ell$ in $G_{\Phi}$ is $n_{i,\ell}$, where $\ell \in \{r,b\}, i \in [w]$. 
	
	If there are no colorings satisfying the constraint  $OPT(\Phi, n_{1,r}, n_{1,b},\ldots, n_{w,r}, n_{w,b},k_1)$, then we set its value equal to -$\infty$.
	Observe that $G$ is a \textsc{Yes}-instance of \lcp{} if and only if $OPT(\Psi,., \ldots,., k_1) \geq k_2$, for some $2$-coloring of $G_{\Psi}$.

	Now we give the details of calculating the values of $OPT(\Phi, )$ at each operation.
	\begin{enumerate}
		\item $\Phi=v(i)$. In this case $G_{\Phi}$ contains one vertex of label $i$ and no edges. Hence $OPT(\Phi, 0, 0,\ldots n_{i,r}=1,0, \ldots 0,k_1=0)=0$ and $OPT(\Phi, 0, 0,\ldots 0, n_{i,b}=1,0, \ldots 0,0,k_1=0)=0$. Otherwise $OPT(\Phi,\ldots, k_1)=- \infty$.
		
		\item  $\Phi=\rho_{i \rightarrow j} (\Phi')$.
		
		All vertices of label $i$ are relabled to $j$ by $\rho_{i \rightarrow j}$ operation in $G_{\Phi}$, hence, there are no vertices of label $i$ in
		$G_\Phi$, so $n_{i,r}=n_{i,b}=0$. 	Let $c$ be a coloring corresponding to an entry $OPT(\Phi, n_{1,r}, n_{1,b},\ldots, n_{w,r}, n_{w,b},k_1)$.
		Then $c$ is also a coloring of $G_{\Phi'}$, corresponding to the entry $OPT(\Phi', n'_{1,r}, n'_{1,b},\ldots, n'_{w,r}, n'_{w,b},k_1)$, where $n'_{i,\ell}+n'_{j,\ell}=n_{j,\ell}$ for  $\ell \in \{r,b\}$ and $n'_{p,\ell}=n_{p,\ell}$ for all $p \in [w] - \{i,j\}$ and  $\ell \in \{r,b\}$. 
		The number of red and blue edges in $G_{\Phi}$ is the same as that in $G_{\Phi'}$ with respect to the coloring $c$ .
		Hence, we have the following relation.

			\begin{multline*}\hspace{-0.1cm} OPT(\Phi, n_{1,r}, n_{1,b},\ldots, n_{w,r}, n_{w,b},k_1)=\\\hspace{-0.1cm} \left\{\begin{array}{l}
			\max
			\left\{
			OPT(\Phi', n'_{1,r},\ldots,n'_{w,b},k_1) 
			\left|
			\begin{matrix}
			n'_{i,\ell}+n'_{j,\ell}=n_{j,\ell}, n_{i,\ell}=0 \text{~for each~} \ell \in \{r,b\} \\
			\text{and~} n_{p,\ell}=n'_{p,\ell} \text{~for all~} p \in [w]-\{i,j\} \\\hspace{-3.6cm} \text{and~} \ell\in \{r,b\}\\
			\end{matrix}
			\right.
			\right\} \\
				-\infty, \text{ otherwise}
			\end{array}
			\right.\end{multline*}

		\item $\Phi= \Phi' \oplus \Phi''$. As this operation does not add new edges,  any coloring $c$ corresponding to $OPT(\Phi, n_{1,r}, n_{1,b},\ldots, n_{w,r}, n_{w,b},k_1)$ is split between two colorings $c'=c|_{V(G_{\Phi'})}$ and  $c''=c|_{V(G_{\Phi''})}$ respectively. 
		As $G_{\Phi}$ is the disjoint union of $G_{\Phi'}$  and $G_{\Phi''}$, the number of red and blue edges edges in $G_{\Phi}$ with respect
		to $c$ is a sum of number red and blue edges with respect to $c'$  and $c''$ in the graphs $G_{\Phi'}$  and $G_{\Phi''}$. 

		\begin{multline*}\hspace{-0.2cm} OPT(\Phi, n_{1,r}, n_{1,b},\ldots, n_{w,r}, n_{w,b},k_1)=\\\hspace{-0.2cm} \max\limits_{\begin{matrix} n'_{i,a}+n''_{i,a}=n_{i,a}\\ k'_1+k''_1 =k_1 \end{matrix}}\hspace{-0.2cm} \left \{OPT(\Phi', n'_{1,r},\ldots, n'_{w,b},k'_1) + OPT(\Phi'', n''_{1,r},\ldots, n''_{w,b},k''_1)\right\}\end{multline*}
		
		\item $\Phi = \eta_{i,j} (\Phi')$.		The graph $G_{\Phi}$ is obtained from  $G_{\Phi'}$ by adding the edges between each vertex of label $i$ to each vertex of label $j$.  Any coloring $c$ of $G_{\Phi}$ is also a coloring of $G_{\Phi'}$.  As given $w$-expression is nice, every edge is which is added by this operation was not present in $G_{\Phi'}$. Therefore $\eta_{i,j}$ operation on $G_{\Phi'}$ creates  $n_{i,r}\cdot n_{j,r}$ many red edges and $n_{i,b}\cdot n_{j,b}$ blue edges. Hence, we have the following relation

		$$\hspace{-0.5cm} OPT(\Phi, n_{1,r}, \ldots, n_{w,b},k_1)=OPT(\Phi', n_{1,r}, \ldots, n_{w,b},k_1-n_{i,r}\cdot n_{j,r})+n_{i,b}\cdot n_{j,b}$$
		
	\end{enumerate}
	We have described the recursive formulas for all possible cases. The correctness of the algorithm follows from the description of the procedure.
	The number of entries in the OPT table is at most $|\Psi| n^{O(w)}$. 
	We can compute each entry of the OPT table in $n^{O(w)}$ time. 
	The maximum number of blue edges that can be obtained in $G_{\Psi}$ is	equals to $\max_{n_{1,r},\ldots,n_{w,b}}OPT(\Psi, n_{1,r}, n_{1,b},\ldots, n_{w,r}, n_{w,b},k_1)$ which can be computed in $n^{O(w)}$ time. This proves that \lcp{} can be solved in time $n^{O(w)}$ on graphs of clique-width at most $w$. \qed
\end{proof}


\subsection{Lower Bound}
We now show the lower bound complementing with the corresponding upper bound result of the previous section. To prove our result we give a linear $\fpt$ reduction from the \textsc{Minimum Bisection} problem.
In the  \textsc{Minimum Bisection} problem, we are given a graph $G$ with an even number
of vertices and a positive integer $k$, and the goal is to determine whether there is a partition of $V(G)$
into two sets $V_1$ and $V_2$ of equal size such that $|E_G[V_1 , V_2]|\leq k$. Fomin et al.~\cite{fomin2010algorithmic} showed that  \textsc{Minimum Bisection} cannot be solved in time $f(w)n^{o(w)}$ unless ETH fails.

\begin{theorem}\label{thm-cw-lower}
	The \lcp{} problem cannot be solved in time $f(w)n^{o(w)}$ unless ETH fails. Here, $w$ is the clique-width of $n$ vertex input graph.
\end{theorem}
\begin{proof}
	We give a reduction from the \textsc{Minimum Bisection} problem
	to the \lcp{} problem. Let $(G,k)$ be an instance of \textsc{Minimum Bisection}. We construct a graph $H$ as follows. 
	
	\begin{enumerate}
		\item For every vertex $v \in V(G)$, we introduce two vertices $a_v, b_v$ in $H$. Let $A=\{a_{v_1}, \ldots, a_{v_n}\}$ and $B=\{b_{v_1}, \ldots, b_{v_n}\}$
		\item For every edge $uv \in E(G)$, we add the edges $a_ua_v$ and $b_ub_v$ to $H$.
		\item Finally, for every vertex $a_{v_i} \in A$ and every vertex $b_{v_j} \in B$ add the edge $a_{v_i}b_{v_j}$ to $H$.  
	\end{enumerate}
	It is easy to see that the graph $H$ has $2n$ vertices and $2m+n^2$ edges and the construction of $H$ can be done in polynomial time. Moreover, if $cw(G)=w$ with $w$-expression $\Phi_G$, then we can construct a $(w+1)$-expression $\Phi_H$ of $H$ by taking two disjoint copies of $\Phi_G$ and relabel every vertex in the first copy with the label $w+1$ and every vertex in the second copy with some arbitrary label $\ell \in [w]$ and finally add edges between both the copies using $\eta_{w+1,\ell}$. This shows that $cw(H) \leq w+1$.
	Let us set
	$k_1=k_2=m-k+n^2/4$. We now show that $(G,k)$ is a \textsc{Yes} instance of \textsc{Minimum Bisection} if and only if $(H,k_1,k_2)$ is a
	\textsc{Yes} instance of \lcp{}.
	
	\paragraph{Forward direction.} Let $(V_1,V_2)$ be a partition of $V(G)$ such that $|E_G[V_1,V_2]| \leq k$ and $|V_1|=|V_2|$. We construct a $2$-coloring $f:V(H) \rightarrow \{r,b\}$ of $H$ as follows. For each $v \in V(G)$, $f(a_v)=r$ and $f(b_v)=b$ if $v \in V_1$ and $f(a_v)=b$ and $f(b_v)=r$ if $v \in V_2$.
	Let $A_r=\{a_v: f(a_v)=r\}, A_b=\{a_v: f(a_v)=b\}$ and $B_r=\{b_v: f(b_v)=r\}, B_b=\{b_v: f(b_v)=b\}$. It is easy to see that $|A_r|=|A_b|=|B_r|=|B_b|=n/2$.
	$$r_f=|E_H[A_r]|+|E_H[B_r]|+|E_H[A_r,B_r]|=m-k+n^2/4 = k_1$$
	
	Similarly, we can show that $b_f = k_2$. Therefore $(H,k_1,k_2)$ is a \textsc{Yes} instance of \lcp{}.
	
	\paragraph{Reverse direction.} Let $f:V(H) \rightarrow \{r,b\}$ be a $2$-coloring of $H$ such that $r_f = k_1$ and $b_f = k_2$.
	Let $A_r=\{a_v: f(a_v)=r\}, A_b=\{a_v: f(a_v)=b\}$ and $B_r=\{b_v: f(b_v)=r\}, B_b=\{b_v: f(b_v)=b\}$.
	Let $V_r=A_r \cup B_r$
	and $V_b=A_b \cup B_b$. Then we have $|E_H[V_r,V_b]|=2m+n^2-r_f-b_f=2k+n^2/2$.
	
	Let $V_1:=\{v : f(a_v)=r\}$ and $V_2:=\{v: f(a_v)=b\}$.
	We show that $|E_G(V_1,V_2)|\leq k$  and $|V_1|=|V_2|$.
	Let $|A_b|=p$ and $|B_b|=q$ such that $p+q=|V_b|$. Then we have $|A_r|=n-p$ and $|B_r|=n-q$.
	
	We know that $$|E_H[V_r]|+|E_H[V_b]|=k_1+k_2=2m-2k+n^2/2$$
	$$|E_H[A_r]|+|E_H[B_r]|+|E_H[A_r,B_r]|+|E_H[A_b]|+|E_H[B_b]|+|E_H[A_b,B_b]|=2m-2k+n^2/2$$
	$$|E_H[A_r]|+|E_H[B_r]|+(n-p)(n-q)+|E_H[A_b]|+|E_H[B_b]|+pq=2m-2k+n^2/2$$

	After simplifying we get $(n-p)(n-q)+pq=n^2/2$.
	Which implies $p=q=n/2$, that is $|A_r|=|A_b|=|B_r|=|B_b|=n/2$.
	Hence  $|V_1|=|V_2|=n/2$, that is ($V_1,V_2)$ is a bisection of $G$.

	Finally, we show that $|E_G[V_1,V_2]| \leq k$. 
	We know that $|E_H[V_r,V_b]|=2k+n^2/2$. $$|E_H[V_r,V_b]|=|E_H[A_r,A_b]|+|E_H[B_r,B_b]|+|E_H[A_r,B_b]|+|E_H[A_b,B_r]|=2k+n^2/2$$
	By the construction of $H$ we have $|E_H[A_r,A_b]|=|E_H[B_r,B_b]|$. Therefore we get
	$$2|E_H[A_r,A_b]|+n^2/4+n^2/4=2k+n^2/2$$
	Which implies $|E_H[A_r,A_b]|=k$ and hence $|E_G[V_1,V_2]|=k$.
	Therefore $(G,k)$ is a \textsc{Yes} instance of \textsc{Minimum Bisection}. This concludes the proof.
\end{proof}

\section{Parameterized Algorithms}\label{sec-palgo}

Clique-width of cluster graphs, co-cluster graphs and threshold graphs is at most two. Hence, from the Theorem~\ref{thm-cw-xp}, \lcp{} is  polynomial time solvable on these graph classes. In this section, we show that \lcp{} is \fpt parameterized by distance to cluster graphs, distance to co-cluster graphs and distance to threshold graphs.

\subsection{Distance to Cluster Graphs}
\begin{theorem}\label{thm-dist-cluster-1}
	\lcp{} is fixed-parameter tractable parameterized by the distance to cluster graphs.
\end{theorem}
\begin{proof}
	
	%
	Let $(G, X, k_1,k_2)$ be a \lcp{} instance, where $X \subseteq V(G)$ of size $k$ such that $G-X$ is a disjoint union of cliques $C_1,C_2,\ldots, C_{\ell}$.
	We first guess the colors of vertices in $X$ in an optimal $2$-coloring of $G$. This can be done in $O(2^{k})$ time. Let $h: X \rightarrow \{r,b\}$  be such a coloring. Without loss of generality we assume that $X$ is an independent set in $G$. Otherwise, let $e_r$ and $e_b$ be the number of red and blue edges in the coloring $h$ having both end vertices in the set $X$.
	We build the new instance $(G',X,k_1',k_2')$ of \lcp{}, where $G'$ is the graph obtained from $G$ by deleting the edges having both their end vertices in $X$ and $k_1'=k_1-e_r$ and $k_2'=k_2-e_b$. It is easy to see that $(G,X,k_1,k_2)$ is a \textsc{Yes} instance of \lcp{} iff there exists a $2$-coloring $g$ of $G'$ such that $g|_X=h$ and $r_g \geq k'_1$ and $b_g \geq k'_2$. Hence, we  assume that $X$ is an independent set in $G$.
	
	 For each $i \in [\ell]$, let $G_i=G[X \cup C_i]$ be the subgraph of $G$ induced by the vertices of the clique $C_i$ and the set $X$.

	The algorithm is based on dynamic programming technique, which has two phases. 
	In phase-1, given a graph $G_i$ and non-negative integers $q$ and $n^r_i$, we find a $2$-coloring of $G_i$ that maximizes the number of blue edges with the constriant that there are at least $q$ red edges in $G_i$ and $n^r_i$ red vertices in $C_i$. In phase-2, for each $t \in [\ell]$, and a non-negative integer $p$, we find a $2$-coloring of $\widehat{G_t}=G[C_1\cup,\ldots,\cup C_t\cup X]$ that maximizes the number of blue edges with the constraint that number of red edges is at least $p$.

	\paragraph{\bf Phase-1.} 	For each $i \in [\ell]$, $q \in [|E(G_i)|]\cup\{0\}$ and $n^r_i \in [|C_i|]\cup\{0\}$, let $b[G_i,n^r_i,q]$ be the maximum number of blue edges that can be attained in a $2$-coloring $g$ of $G_i$ satisfying the following constraints. 
	\begin{enumerate}
		\item $g|_{X}=h$.
		\item number of red edges in $G_i$ is at least $q$.
		\item number of red vertices in $C_i$ is equal to $n^r_i$. 
	\end{enumerate}
	 If the constraint cannot be satisfied, then we let $b[G_i,n^r_i,q]=-\infty$. 
	From the definition of $b[,]$,  we can see that $b[G_i,0,0]$ gives the number of blue edges in $G_i$ when all vertices of $C_i$ are colored blue. $b[G_i,0,q]=-\infty$ for $q >0$, $b[G_i,n^r_i,0]=-\infty$ for $n^r_i >1$.
	For a given values of $i$, $n^r_i$ and $q$, the computation of $b[G_i,n^r_i,q]$ is  described as follows.

	Let $H_i$ be the graph obtained from $G_i$ by deleting the edges inside the clique $C_i$. It is easy to see that $X$ is a vertex cover of the graph $H_i$. If $g$ colors $n^r_i$ vertices red in the clique $C_i$ then we get ${n^r_i \choose 2}$ red edges and ${n^b_i \choose 2}$
	blue edges inside $C_i$, where $n^r_i+n^b_i=|C_i|$. Hence we get the following relation. 
	
	$$b\Big[G_i,n^r_i,q\Big]=b\Big[H_i,n^r_i,q-{n^r_i \choose 2}\Big] +{n^b_i \choose 2}$$
	For $v \in H_i-X$, let $r^i(v)$ and $b^i(v)$ denote the number of red and blue vertices in $N(v) \cap X$ respectively. For a vertex $v \in H_i$, we use $H_i-\{v\}$ to denote the graph obtained from $H_i$ by deleting the vertex $v$ and its incident edges. Let $q'=q-{n^r_i \choose 2}$. Using this notation, we get the following recurrence. 
	
	$$b\Big[H_i,n^r_i,q'\Big]= \max \Big\{b\big[H_i-\{v\},n^r_i-1, \max\{q'-r^i(v), 0\}\big], b\big[H_i-\{v\}, n^r_i, q'\big]+b^i(v)\Big\}$$
	
	If $v$ is colored red, then we get $r^i(v)$ red edges between $v$ and the neighbors of $v$ in $X$. If $v$ is colored blue, then we get  $b^i(v)$ blue edges between $v$ and neighbors of $v$ in $X$.
	
	The size of the DP table is at most $O(mn^2)$ and each entry can be computed in $O(n)$ time. Hence the running time of Phase-1 is $O(n^3m)$.
	
	\paragraph{\bf Phase-2.} Let $\widehat{G_t}$ be the subgraph of $G$ induced by the cliques $C_1,\cdots,C_t$ and the set $X$.
	Let $OPT[t,p]$ be the maximum number of blue edges that can be attained in a $2$-coloring $f$ of $\widehat{G_t}$ with the constraint that number of red edges is at least $p$ and $f|_{X}=h$. If the constraint cannot be satisfied, then we let $OPT[t,p]=-\infty$.
	From the definition of $OPT$, we have $OPT[0,0]=0$ and $OPT[0,p]=-\infty$ for $p >0$. For $t>0$ we have:
	
	$$OPT[t,p] = \max\limits_{\begin{matrix} q=0, \ldots, |E(G_t)|\\ n^r_t=0, \ldots, |C_t| \end{matrix}}  \Big\{OPT\big[t-1, \max\{p-q, 0\}\big]+b[G_t,n^r_t,q]\Big\}$$

	If there are $q$ red edges in $G_t$ and $n^r_t$ red vertices in $C_t$ in the coloring $f$, then we get $b[G_t,n^r_t,q]$ blue edges in $G_t$.   We consider all possible values for  $q$ and $n^r_t$ and pick the values that maximizes the $OPT$.
	
	Observe that $(G,X,k_1,k_2)$ is a \textsc{Yes} instance if and only if $OPT[\ell,k_1] \geq k_2$.  There are $O(\ell k_1)$ subproblems, each of which can be solved in $O(n^2m)$ time. As $\ell\leq n, k_1 \leq m$, the  running time of phase-2 is $O(\ell k_1n^2m)=O(n^3m^2)$.
	The overall running time of the algorithm is $O(2^kn^3m^2)$, where $O(2^k)$ is the time required for guessing the coloring of $X$ in an optimal coloring of $G$.  
	
\end{proof}

\subsection{Distance to Co-cluster Graphs}
\begin{theorem}\label{thm-dist-cocluster-1}
	\lcp{} is fixed-parameter tractable parameterized by the distance to co-cluster graphs.
\end{theorem}
	\begin{proof}
	Let $(G, X, k_1,k_2)$ be a \lcp{} instance, where $X \subseteq V(G)$ of size $k$ such that $G - X$  is a co-cluster graph. The graph $G-X$ is either an independent set or a connected graph. If $G-X$ is an independent set then $X$ is a vertex cover of $G$, hence we apply the algorithm of~\cite{gutin2014parameterized} to solve this instance in $\fpt$ time.
	
	Hence, we assume that $G-X$ is connected. As $G-X$ is a complement of a cluster graph we can partition the co-cluster $G-X$ into maximal independent sets $I_1, I_2, \ldots, I_{\ell}$.  Note that for any $i,j \in [\ell]$ and $i \neq j$, every vertex of $I_i$ is adjacent to every vertex of $I_j$. We guess the colors of vertices in $X$ in an optimal $2$-coloring of $G$. This can be done in $O(2^{k})$ time. Let $h: X \rightarrow \{r,b\}$  be such a coloring. 
	We assume that $X$ is an independent set in $G$, otherwise we count the number of red and blue edges inside $X$ and adjust the parameters $k_1$, $k_2$ and delete all the edges inside $X$. 
	
	For each $i \in [\ell]$, let $G_i=G[X \cup I_i]$ be the subgraph of $G$ induced by the vertices of the independent set $I_i$ and the set $X$.

	The algorithm is similar to the algorithm of distance to cluster described in Theorem~\ref{thm-dist-cluster-1}. It is based on dynamic programming technique, which has two phases. 
	In phase-1, given a graph $G_i$ and non-negative integers $q$ and $n^r_i$, we find a $2$-coloring of $G_i$ that maximizes the number of blue edges with the constriant that there are at least $q$ red edges in $G_i$ and $n^r_i$ red vertices in $I_i$. In phase-2, for each $t \in [\ell]$, and a non-negative integer $p$, we find a $2$-coloring of $\widehat{G_t}=G[I_1\cup,\ldots,\cup I_t\cup X]$ that maximizes the number of blue edges with the constraint that number of red edges is at least $p$.

	\paragraph{\bf Phase-1.} We can observe that $X$ is a vertex cover of $G_i$ of size $k$.	For each $i \in [\ell]$, $q \in [|E(G_i)|] \cup \{0\}$ and $n^r_i \in [|I_i|] \cup \{0\}$, let $b[G_i,n^r_i,q]$ be the maximum number of blue edges that can be attained in a $2$-coloring $g$ of $G_i$ satisfying the following constraints. 
	\begin{enumerate}
		\item $g|_{X}=h$.
		\item number of red edges in $G_i$ is at least $q$.
		\item number of red vertices in $I_i$ is equal to $n^r_i$. 
	\end{enumerate}
	If the constraint cannot be satisfied, then we let $b[G_i,n^r_i,q]=-\infty$. 
	From the definition of $b[,]$,  we can see that $b[G_i,0,0]$ gives the number of blue edges in $G_i$ when all vertices of $I_i$ are colored blue. $b[G_i,0,q]=-\infty$ for $q >0$, $b[G_i,n^r_i,0]=-\infty$ for $n^r_i >1$.
	For $v \in G_i-X$, let $r^i(v)$ and $b^i(v)$ denote the number of red and blue vertices in $N(v) \cap X$ respectively. 
	For a given values of $i$, $n^r_i$ and $q$, the computation of $b[G_i,n^r_i,q]$ is  described as follows.

	$$b\Big[G_i,n^r_i,q\Big]= \max \Big\{b\big[G_i-\{v\},n^r_i-1, \max\{q-r^i(v), 0\}\big], b\big[G_i-\{v\}, n^r_i, q\big]+b^i(v)\Big\}$$
	
	If $v$ is colored red, then we get $r^i(v)$ red edges between $v$ and the neighbors of $v$ in $X$. If $v$ is colored blue, then we get  $b^i(v)$ blue edges between $v$ and neighbors of $v$ in $X$.
	
	The size of the DP table is at most $O(mn^2)$ and each entry can be computed in $O(n)$ time. Hence the running time of Phase-1 is $O(n^3m)$.

	\paragraph{\bf Phase-2.} Let $\widehat{G_t}$ be the subgraph of $G$ induced by the independent sets $I_1,\cdots,I_t$ and the set $X$.
	Let $OPT[t,N^r_t,p]$ be the maximum number of blue edges that can be attained in a $2$-coloring $f$ of $\widehat{G_t}$ satisfying the following constrains.
	
	\begin{enumerate}
		\item $f|_{X}=h$.
		\item number of red edges in $\widehat{G_t}$ is at least $p$.
		\item number of red vertices in $\widehat{G_t}-X$ is $N^r_t$.
	\end{enumerate}
	If the constraint cannot be satisfied, then we let $OPT[t,N^r_t,p]=-\infty$.
	From the definition of $OPT$, we have $OPT[0,0,0]=0$ and $OPT[0,N^t_r,p]=-\infty$ for $p >0$,  $OPT[t,0,p]=-\infty$ for $p >0$, For $t>0$ we have:
	
	
	\begin{align}
		OPT[t,N^r_t,p] &= \hspace{-0.3cm} \max\limits_{\begin{matrix} q=0, \ldots, |E(G_t)|\\ n^r_t=0, \ldots, |I_t| \end{matrix}}\Bigg\{OPT\big[t-1, N^r_t-n^r_t, \max\{p-q-n^r_t(N^r_t-n^r_t), 0\}\big] \notag\\ 
		&+b[G_t,n^r_t,q]+(N^b_t-|V(G_t)|+n^r_t)(|V(G_t)|-n^r_t)\Bigg\} \notag \label{eq:energycart}
	\end{align}

	Where $N^b_t= |V(\widehat{G_t})|-N^r_t$.
	If there are $q$ red edges in $G_t$ and $n^r_t$ red vertices in $I_t$ in the coloring $f$, then we get $b[G_t,n^r_t,q]$ blue edges in $G_t$. Also, we get $n^r_t(N^r_t-n^r_t)$ red edges and $(N^t_b-|V(G_t)|+n^r_t)(|V(G_t)|-n^r_t)$ blue edges between the sets $I_t$ and $I_1\cup,\ldots,I_{t-1}$ respectively.   We consider all possible values for number of red edges in $G_t$ and red vertices in $I_t$ and pick the one that maximizes the $OPT$.
	
	Observe that $(G,X,k_1,k_2)$ is a \textsc{Yes} instance if and only if $OPT[\ell,N^{r}_{\ell},k_1] \geq k_2$ for some non-negative integer $N^{r}_{\ell}$.  There are $O(\ell n k_1)$ subproblems, each of which can be solved in $O(n^2m)$ time. As $\ell\leq n, k_1 \leq m$, the overall running time of this algorithm is $O(2^{k}\ell k_1n^3m)=O(2^kn^4m^2)$. 
	
\end{proof}

\subsection{Distance to Threshold Graphs}
\begin{theorem}\label{thm-dist-threshold}
	\lcp{} is fixed-parameter tractable parameterized by the distance to threshold graphs.
\end{theorem}
\begin{proof}
	Let $(G, X, k_1,k_2)$ be a \lcp{} instance, where $X \subseteq V(G)$ of size $k$ such that $G-X$ is a threshold graph.  We guess the coloring of $X$ in an optimal $2$-coloring  of $G$ in $O(2^k)$ time. Let $h:X \rightarrow \{r,b\}$ be such a coloring. We assume that $X$ is an independent set in $G$, if not we count the number of red edges $e_r$ and blue edges $e_b$ inside $X$ and replace the parameters $k_1$ and $k_2$ with $k_1-e_r$ and $k_2-e_b$ respectively. Then finally we delete all the edges inside $X$.   For a vertex $v$ in $G-X$, we use $n^r_X(v)$ (resp. $n^b_X(v)$) to denote the number of neighbors of the vertex $v$ in the set $X$ which are colored red (resp. blue).  Let $v_1,v_2,\ldots,v_{\ell}$ be the ordering of the vertices of the threshold graph $G-X$  obtained from its construction (i.e., $v_i$ is added before $v_{i+1}$ to the graph $G-X$).

	For $t \in [\ell]$, let $V_t=\{v_1,\ldots,v_t\}$ and  $G_t$ be the graph induced by the vertices $V_t \cup X$. Using the definition of threshold graphs, we can see that, for each $t \in[\ell]$, the vertex $v_t$ is either a universal vertex or an isolated vertex in the graph $G_t-X$. We use $n^r_t$ and $n^b_t$ denote the number of red and blue vertices in a $2$-coloring of $G_t$ respectively. 
	
	Let $OPT[t,n^r_t,p]$ be the maximum number of blue edges that can be obtained in a $2$-coloring $g$ of $G_t$ with the constraint that $g|_X=h$ and $V_t$ has  $n^r_t$ vertices of color red, $G_t$ has at least $p$ red edges. If the constraint cannot be satisfied, then we let $OPT[t,n^r_t,p]=-\infty$.
	From the definition of $OPT$, we have $OPT[0,0,0]=0$, $OPT[0,n^r_t,p]=-\infty$ for $p >0$ and $OPT[t,0,p]=-\infty$ for $p>0$.

%
	
	For $t>0$, we have two cases based on whether $v_t$ is a universal vertex or an isolated vertex in the graph $G_t-X$. If $v_t$ is a universal vertex, then it is adjacent to all vertices of $G_t-X$. Hence we get the following relation. 
	
	$OPT[t,n^r_t,p] = $
	$$\max \big\{OPT[t-1, n^r_t-1, \max\{p-(n^r_t-1)-n^r_X(v_t),0\}], OPT[t-1, n^r_t, p]+ (n^b_t-1+n^b_X(v_t))\} $$
%

	If $v_t$ is colored red, then we get $(n^r_t-1)$ red edges between $v_t$ and neighbors of $v_t$ in $G_t-X$ and  $n^r_X(v_t)$ red edges between $v_t$ and its neighbors in $X$. 	If $v_t$ is colored blue, then we get $n^b_t-1$ blue edges between $v_t$ and neighbors of $v_t$ in $G_t-X$ and  $n^b_X(v_t)$ blue edges between $v_t$ and its neighbors in $X$. 
	
	If $v_t$ is an isolated vertex, then it is not adjacent to any vertex of  $G_t-X$. Then we get the following relation.
		$$OPT[t,n^r_t,p] = 
	\max \big\{OPT[t-1, n^r_t-1, \max\{p-n^r_X(v_t),0\}], OPT[t-1, n^r_t, p]+ n^b_X(v_t)\} $$
	If $v_t$ is colored red, then we  get $n^r_X(v_t)$ red edges between $v_t$ and its neighbors in $X$. 	If $v_t$ is colored blue, then we  get   $n^b_X(v_t)$ blue edges between $v_t$ and its neighbors in $X$.
	
	Observe that $(G,X,k_1,k_2)$ is a \textsc{Yes} instance of \lcp{} if and only if $OPT[\ell,n^r_t, k_1] \geq k_2$ for some integer $n^r_t$.  There are $O(n^2 k_1)$ subproblems, each of which can be solved in $O(n)$ time. As $k_1 \leq m$, the overall running time of this algorithm is $O(2^kn^3m)$.  
	
\end{proof}	

\section{Special Graph Classes}~\label{sec-bipartite}
In this section, we show that \lcp{} is $\NP$-complete on bipartite graphs and split graphs. 

\begin{theorem}\label{thm-bipartite}
	\textsc{ $2$-Load coloring} is $\NP$-complete on bipartite graphs.
\end{theorem}
\begin{proof}
	We give a reduction from \lcp{} on general graphs. 
	Let $(G,k_1,k_2)$ be an instance of \lcp{}. Without loss of generality, we assume that $k_2 \geq k_1$.  We construct a bipartite graph $H = ((X\cup Z)\cup Y),E)$ as follows. For every vertex $v \in V(G)$, we introduce a vertex $x_v \in X$. For every edge $e \in E(G)$, we introduce a vertex $y_e \in Y$, and if $e =uv$, then $y_e$ is adjacent to $x_u$ and $x_v$ in $H$. For every edge $e \in E(G)$, we also introduce additional $m^2$ vertices $z^1_e,z^2_{e},\ldots,z^{m^2}_e$ in $Z$. For each $e \in E(G)$ and  $i \in [m^2]$ introduce an edge between $z^i_e$ and $y_e$ in $H$.  This completes the construction of the bipartite graph $H$ and clearly can be performed in polynomial time. We set $k_1^\prime = 2k_1+k_1m^2$, $k_2^\prime=(m+k_2-k_1)+(m-k_1)m^2$. We argue that $(G,k_1,k_2)$ is a \textsc{yes} instance of \lcp{} if and only if $(H,k_1^\prime,k_2^\prime)$ is a \textsc{yes} instance of \lcp{}. 
	
	
	\emph{Forward direction}. Let $f: V(G) \rightarrow \{r,b\}$ is a  $2$-coloring of $G$ such that $r_f=k_1$ and $b_f =k_2$.  Then we define a coloring $g: V(H) \rightarrow \{r,b\}$ of $H$ as follows:  $g(x_v) = f(v)$ for all $x_v \in X$. For every edge $e = uv \in E(G)$, if $f(u)=f(v)=r$, then $g(y_e) = f(u)$, else $g(y_e)=b$. For each $e \in E(G)$ and $i \in [m^2]$, $g(z^i_e)=g(y_e)$. 
	Note that for each red (resp blue) edge $e=uv$ in $G$, $m^2+2$ edges (namely $y_ex_v$, $y_ex_u$, $z^i_ex_u$) are red (resp. blue) with respect to $g$. For each edge $e=uv \in E(G)$ with $f(u) \neq f(v)$ we get $m^2+1$ blue edges with respect to $g$.  Therefore, $r_g=(m^2+2)k_1=2k_1+m^2k_1$ and $b_g=(m^2+2)k_2+(m-k_1-k_2)(m^2+1)=(m+k_2-k_1)+(m-k_1)m^2$. Hence $(H,k_1^\prime,k_2^\prime)$ is a \textsc{yes} instance of \lcp{}.
	
	\emph{Reverse direction.} Let  $g: V(H) \rightarrow \{r,b\}$ be a $2$-coloring of $H$ such that $r_g=2k_1+k_1m^2$ and $b_g=(m+k_2-k_1)+(m-k_1)m^2$. We assume that $g(z^i_e)=g(y_e)$ for each $e \in E(G)$ and $i \in [m^2]$, otherwise we recolor the vertices $z^i_e$ with the color of $y_e$, as it only increases the number of red and blue edges.  Now, we argue that the $2$-coloring $g$ of $H$ when restricted to the vertices of $G$ gives a $2$-coloring $f$ of $G$ such that $r_f=k_1$ and $b_f=k_2$.  
	
	We partition the vertices of $Y$ in $H$ into four sets based on their coloring in $g$ as follows.
	$$P_r=\{y_e: e=uv, ~g(y_e)=g(x_u)=g(x_v)=r\}$$
	$$Q_r=\{y_e: e=uv, ~g(y_e)=r\text{~ and~} g(x_u)\neq g(x_v)\}$$
	$$P_b=\{y_e: e=uv, ~g(y_e)=g(x_u)=g(x_v)=b\}$$
	$$Q_b=\{y_e: e=uv, ~g(y_e)=b\text{~and~} g(x_u)\neq g(x_v)\}$$
	
	Let $p_r$, $q_r$,$p_b$ and $q_b$ denote the sizes of the sets $P_r$, $Q_r$ $P_b$ and $Q_b$ respectively.  Using this notation, we have  $r_g=2p_r+q_r+m^2(p_r+q_r)$, $b_g=2p_b+q_b+m^2(p_b+q_b)$. Also, it is easy to see that $p_r+q_r+p_b+q_b=m$. By using above two equations, we get $$r_g+b_g=p_r+p_b+m+m^3$$
	$$p_r+p_b=r_g+b_g-m-m^3 = 2k_1+k_1m^2+(m+k_2-k_1)+(m-k_1)m^2-m-m^3=k_1+k_2$$

	\noindent\emph{Claim.} $p_r=k_1$ and $p_b=k_2$.
	
	Suppose $p_r=k_1 - \ell$  and $p_b=k_2+\ell$ for some $\ell \geq 1$. Then $$r_g=2p_r+q_r+m^2(p_r+q_r)=2(k_1-\ell)+q_r+m^2(k_1-\ell+q_r)$$ 
	Since $r_g=2k_1+k_1m^2$, by substituting in the above equation, we get $q_r=\Big(\frac{m^2+2}{m^2+1}\Big)\ell \geq \ell+1$. That implies $q_b \leq m-(k_1-\ell) -(k_2+\ell)-(\ell+1)=(m-k_1-k_2-\ell-1)$. Then
	\begin{equation*}
	\begin{aligned}
	b_g ={} &2p_b+q_b+m^2(p_b+q_b) \\
	=& 2(k_2+\ell)+(m-k_1-k_2-\ell-1)+m^2(k_2+\ell+(m-k_1-k_2-\ell-1)) \\
	<& (m+k_2-k_1)+m^2(m-k_1)
	\end{aligned}
	\end{equation*}
	
	This is a contradiction as $b_g=(m+k_2-k_1)+m^2(m-k_1)$. Therefore $p_r=k_1$ and $p_b=k_2$. Define $f: V(G) \rightarrow \{r,b\}$ as $f(u)=g(x_u)$. Since $p_r=k_1$ and $p_b=k_2$, we get $r_f\geq k_1$ and $b_f \geq k_2$.  Hence, by restricting the coloring $g$ of $H$ to the vertices of $G$, we get  a coloring $f$ of $G$ with at least $k_1$ red edges and $k_2$ blues edges. Therefore $(G,k_1,k_2)$ is a \textsc{yes} instance of \lcp{}.
\end{proof}
\begin{theorem}\label{thm-split}
	\textsc{ $2$-Load coloring} is $\NP$-complete on split graphs.
\end{theorem}
\begin{proof}
	The proof is similar to the case of bipartite graphs except that we add all possible edges between the vertices of $Y$ to make it a clique in the graph $H$. We set $k_1'=2k_1+m^2k_1+{k_1 \choose 2}$  and $k_2'=(m+k_2-k_1)+(m-k_1)m^2+{m-k_1 \choose 2}$. 
\end{proof}
\section{Conclusion}
In this paper, we have studied the parameterized complexity of  \lcp{}. 
We showed that \lcp{} (a) cannot
be solved in time $f(w)n^{o(w)}$, unless ETH fails, (b) can be solved in time  $n^{O(w)}$, where $w$ is the clique-width of the graph.
 We have shown that the problem is \fpt parameterized by (a) distance to cluster graphs  (b) distance to co-cluster graphs and (c) distance to threshold graphs. We also studied the complexity of the problem on special
classes of graphs. We have shown that the problem is $\NP$-complete on bipartite graphs and split graphs. We conclude with the following open questions.

\begin{itemize}
\item As the problem admits a linear kernel with respect to the solution size, it is natural to study the kernelization complexity of the problem with respect to structural graph parameters.


\item Finally, we do not know if \lcp{} can be solved in polynomial time on interval graphs or permutation graphs.
\end{itemize}

\bibliographystyle{splncs03}
\bibliography{myrefs.bib}

\end{document}